\documentclass[proceedings]{stacs}
\stacsheading{2010}{239-250}{Nancy, France}
\firstpageno{239}

\usepackage{amssymb,amsmath,amsthm}
\usepackage{graphicx}
\usepackage{subfigure}
\usepackage[vlined,ruled,linesnumbered]{algorithm2e}
\dontprintsemicolon
\usepackage{xspace}
\usepackage{path}

\newcommand{\savespace}[1]{\vspace{-0ex}}

\graphicspath{{fig/}}

\newcommand{\tthree}{\ensuremath{\mathrm{T}^3}}
\newcommand{\R}{\ensuremath{\mathbb{R}}}
\newcommand{\eps}{\ensuremath{\varepsilon}}
\newcommand{\opt}{\mbox{{\sc opt}}}
\newcommand{\weight}{\mathrm{weight}}

\newcommand{\wrt}{w.r.t.}

\newcommand{\tsp}{{\sc Tsp}}
\newcommand{\tspalpha}{\tsp$(\hspace*{1.5pt}\cdot\hspace*{2pt},\alpha)$}
\newcommand{\apx}{{\sc apx}\xspace}
\newcommand{\np}{{\sc np}\xspace}
\newcommand{\Reals}{\R}

\renewcommand{\leq}{\leqslant}
\renewcommand{\geq}{\geqslant}
\renewcommand{\le}{\leqslant}
\renewcommand{\ge}{\geqslant}
\DeclareMathOperator{\dist}{dist}

\begin{document}

\title[The Traveling Salesman Problem Under Squared Euclidean
Distances]{The Traveling Salesman Problem \\ Under Squared Euclidean Distances}

\author[tue]{M. de Berg}{Mark de Berg}
\address[tue]{Department of Mathematics and Computer Science, TU~Eindhoven,
  the~Netherlands.}
\email{mdberg@win.tue.nl}
\email{f.s.b.v.nijnatten@tue.nl}
\email{gwoegi@win.tue.nl}
\author[tue]{F. van Nijnatten}{Fred van Nijnatten}
\author[vu]{R. Sitters}{Ren\'e Sitters}
\address[vu]{Faculty of Economics and Business Administration,
  VU Amsterdam, the~Netherlands.}
\email{rsitters@feweb.vu.nl}
\author[tue]{G. J. Woeginger}{Gerhard~J.~Woeginger}
\author[wue]{A. Wolff}{Alexander Wolff}
\address[wue]{Lehrstuhl f\"ur Informatik I, Universit\"at
  W\"urzburg, Germany.}
\urladdr{http://www1.informatik.uni-wuerzburg.de/en/staff/wolff\_alexander}


\begin{abstract}
  Let $P$ be a set of points in $\Reals^d$, and let $\alpha \ge 1$ be
  a real number.  We define the distance between two points $p,q\in P$
  as $|pq|^{\alpha}$, where $|pq|$ denotes the standard Euclidean
  distance between $p$ and $q$.  We denote the traveling salesman
  problem under this distance function by
  \tsp($d,\alpha$).  We design a 5-approximation algorithm for
  \tsp(2,2) and generalize this result to obtain an approximation
  factor of $3^{\alpha-1}+\sqrt{6}^{\,\alpha}\!/3$ for $d=2$ and all
  $\alpha\ge2$.

  We also study the variant Rev-\tsp\ of the problem where the
  traveling salesman is allowed to revisit points.  We present a
  polynomial-time approximation scheme for Rev-\tsp$(2,\alpha)$ with
  $\alpha\ge2$, and we show that Rev-\tsp$(d, \alpha)$ is \apx-hard if
  $d\ge3$ and $\alpha>1$.  The \apx-hardness proof carries over to
  \tsp$(d, \alpha)$ for the same parameter ranges.
\end{abstract}

\maketitle

%%%%%%%%%%%%%%%%%%%%%%%%%%%%%%%%%%%%%%%%%%%%%%%%%%%%%%%%%%%%%%%%%%%%%%%
\section{Introduction}
%%%%%%%%%%%%%%%%%%%%%%%%%%%%%%%%%%%%%%%%%%%%%%%%%%%%%%%%%%%%%%%%%%%%%%%

Motivated by a power-assignment problem in wireless networks (see
below for a short discussion of this application)
Funke~et~al.~\cite{flnl-papwc-08} studied the following special case
\tsp($d,\alpha$) of the Traveling Salesman Problem (\tsp) which is
specified by an integer $d\ge2$ and a real number $\alpha>0$.  The
cities are $n$ points in $d$-dimensional space $\Reals^d$, and the
distance between two points $p$ and $q$ is $|pq|^\alpha$, where $|pq|$
denotes the standard Euclidean distance between $p$ and $q$.
\begin{itemize}
\item The objective in problem \tsp($d,\alpha$) is to find a shortest
  tour (under distances $|\cdot|^\alpha$) that visits every city
  \emph{exactly} once.
\item In the closely related problem Rev-\tsp$(d,\alpha)$, the
  objective is to find a shortest tour that visits every city \emph{at
    least} once; thus the salesman is allowed to revisit
  cities.
\end{itemize}
Note that \tsp$(2,1)$ is the classical two-dimensional Euclidean \tsp\
and that \tsp$(d,\infty)$ is the so-called
\emph{bottleneck \tsp} in $\Reals^d$, where the goal is to find a tour
whose longest edge has minimum length.  We are, however, mainly
interested in the case where $\alpha$ is some small constant, and we
will not touch the case $\alpha=\infty$.

\paragraph{\em Similarities and differences to the classical Euclidean
  TSP}

The classical Euclidean {\tsp} is \np-hard even in two dimensions, but
it is relatively easy to approximate.  In particular, it admits a
polynomial-time approximation scheme: Given a parameter $\eps>0$ and a
set of $n$ points in $d$-dimensional Euclidean space, one can find in
$2^{(d/\eps)^{O(d)}} + (d/\eps)^{O(d)} n\log n$ time a tour whose
length is at most $1+\eps$ times the optimal length~\cite{rao-smith}.

A crucial property of the Euclidean {\tsp} is that the underlying
Euclidean distances satisfy the triangle inequality.  The triangle
inequality implies that no reasonable salesman would ever revisit the
same city: Instead of returning to a city, it is always cheaper to
skip the city and to travel directly to the successor city.  All
positive approximation results for the Euclidean {\tsp} rely heavily
on the triangle inequality.  In strong contrast to this, for exponents
$\alpha>1$ the distance function $|\cdot|^\alpha$ does not satisfy the
triangle inequality.  Thus the combinatorial structure of the problem
changes significantly---for example, revisits may suddenly become
helpful---and the existing approximation algorithms for Euclidean
{\tsp} cannot be applied.

Another nice property of the classical Euclidean problem \tsp$(2,1)$
is that, sloppily speaking, instances with many cities have long
optimal tours.  Consider for instance a set $P$ of $n$ points in the
unit square.  Then there exists a tour whose Euclidean length is
bounded by $O(\sqrt{n})$~\cite{KarpSteele}.  This bound is essentially
tight since there are point sets for which \emph{every} tour
has Euclidean length $\Omega(\sqrt{n})$.  Interestingly, these results
do not carry over to
\tsp$(2,2)$ with \emph{squared}
Euclidean distances.  Problem \#124 in the book by
Bollob\'as~\cite{Bollobas} shows that there always exists a tour for
$P$ such that the sum of the squared Euclidean distances is bounded
by~$4$, and that this bound of~$4$ is best possible.
% (take for instance one point in one corner of the unit square, and a
% cloud of $n-1$ points very close to the diametrical corner).
Since, as a rule of
thumb, large objective values are easier to approximate than small
objective values, this already indicates a substantial difference in
the approximability behaviors of \tsp(2,1) and \tsp(2,2).

\paragraph{\em Previous work and our results}

Funke~et~al.~\cite{flnl-papwc-08} note that the distance function
$|\cdot|^\alpha$ satisfies the so-called \emph{$\tau$-relaxed triangle
  inequality} with parameter $\tau=2^{\alpha-1}$ (see
Section~\ref{sec:tspalpha} for a definition).  The classical TSP under
the $\tau$-relaxed triangle inequality has been extensively
studied~\cite{tsp_relaxed_improved, tsp_relaxed,
  tsp_relaxed_improved_further, bhksu-tnsah-02}, and all the corresponding
machinery from the literature can be applied directly to
\tsp($d,\alpha$).  For instance,
Andreae~\cite{tsp_relaxed_improved_further} derives a
$(\tau^2+\tau)$-approximation for the classical {\tsp} under the
$\tau$-relaxed triangle inequality ($\Delta_\tau$-\tsp, for short).
This result translates into a
$(4^{\alpha-1}+2^{\alpha-1})$-approximation for \tspalpha.  For
$\tau > 3$, it is better to apply Bender and
Chekuri's $4\tau$-approximation~\cite{tsp_relaxed_improved} for
$\Delta_\tau$-\tsp, which yields a
$2^{\alpha+1}$-approximation for \tspalpha.  Funke et al.\ derive a
$(2\cdot 3^{\alpha-1})$-approximation algorithm for \tspalpha, which
for the range $2<\alpha < \log_{3/2} 3 \approx 2.71$ is better than
applying the known
results~\cite{tsp_relaxed_improved_further,tsp_relaxed_improved}.
The best result for $\alpha<2$ is obtained by B\"ockenhauer et
al.~\cite{bhksu-tnsah-02} whose Christofides-based
$(3\tau^2/2)$-approximation for $\Delta_\tau$-\tsp\ yields a $(3 \cdot
2^{2\alpha-3})$-approximation for \tspalpha.

We will demonstrate in Section~\ref{sec:tspalpha} that essentially
\emph{every} variant of the original \tthree-algorithm by Andreae and
Bandelt~\cite{tsp_relaxed} already gives a $(2\cdot
3^{\alpha-1})$-approximation for \tsp$(d,\alpha)$.  The bottom-line of
all this, and the actual starting point of our paper, is that the
machinery around the $\tau$-relaxed triangle inequality only yields a
bound of roughly $2\cdot 3^{\alpha-1}$.  This raises the following
questions: How much can geometry help us in getting even better
approximation ratios?  Can we beat the 6-approximation for \tsp$(2,2)$
of Funke et al.?  We answer these questions affirmatively: We develop a
new variant of the \tthree-algorithm which we call the \emph{geometric
  \tthree-algorithm}.  An intricate analysis in Section~\ref{sec:two}
shows that this yields a 5-appro\-xi\-ma\-tion for \tsp$(2,2)$.  We
then extend our analysis to \tsp$(2,\alpha)$ with~$\alpha>2$, and thus
obtain a $(3^{\alpha-1}+\sqrt{6}^{\,\alpha}\!/3)$-approximation; see
Section~\ref{sec:greatertwo}.  This new bound is always better than
the bound $2\cdot 3^{\alpha-1}$ of Funke et al.\ and of our analysis
of the \tthree-algorithm.

Finally, in Section~\ref{sec:revtsp}, we turn our attention to the
following two questions: (a)~How does the approximability of \tsp\
behave when we make~$\alpha$ larger than one? (b)~Does allowing
revisits change the complexity or the approximability of the problem?
As we know, classical Euclidean \tsp\ (that is, \tsp$(d,1)$) is
\np-hard~\cite{papadimitriou1977ets} and has a polynomial-time
approximation scheme (PTAS) in any fixed number~$d$ of
dimensions~\cite{arora}.  On the other hand, Rev-\tsp$(d,\alpha)$
has---to the best of our knowledge---not been studied before.
Concerning question~(b), complexity behaves as expected:
Rev-\tsp$(d,\alpha)$ is NP-hard for any $d\ge 2$ and any \mbox{$\alpha>0$},
and our (straightforward) hardness argument also works for
\tsp$(d,\alpha)$.
In terms of approximability, we show that whereas the two-dimensional
problem Rev-\tsp$(2,\alpha)$ still has a PTAS for all values
$\alpha\ge2$, the problem becomes \apx-hard for all $\alpha>1$ in
three dimensions.  We were surprised that the \apx-hardness proof,
too, carried over to \tsp$(3,\alpha)$ for all $\alpha>1$.  This
inapproximability result stands in strong contrast to the behavior of
the classical Euclidean \tsp\ (the case $\alpha=1$).

\paragraph{\em The connection to wireless networks}

Consider a wireless network whose nodes are equipped with
omni-directional antennas.  The nodes are modeled as points in the
plane, and every node can communicate with all other nodes that are
within its transmission radius.  The power (that is, the energy)
needed to achieve a transmission radius of $r$ is roughly proportional
to $r^\alpha$ for some real parameter $\alpha$ called the
\emph{distance-power gradient}.  Depending on environmental
conditions, $\alpha$ typically is in the range
2~to~6~\cite[Chapter~1]{g-hawc-01}.  The goal is to assign powers to
the nodes such that the resulting network has certain desirable
properties, while the overall power consumption is minimized.  A
widely studied variant
has the objective to make the resulting network strongly
connected~\cite{symmetric_fork,fuchs_hardness,kirousis_mst}.  Other
variants (finding broadcast trees; having small hop diameter; etc)
have been studied as well.
Funke~et~al.~\cite{flnl-papwc-08} suggest that it is useful to have a
tour through the network, which can be used to pass a \emph{virtual
  token} around.  The resulting
power-assignment
problem is \tsp$(2,\alpha)$.

Another setting related to \tsp$(2,\alpha)$ is the following.  Instead
of omni-directional antennas, some wireless networks use directional
antennas.  This achieves the same transmission radius under a smaller
energy consumption~%\cite{Huang02topologycontrol}
\cite{nasipuri,Ramanathan01onthe}.
To model directional antennas, Caragiannis~et~al.~\cite{directional}
assume that a node can communicate with other nodes in a circular
sector of a given angle (where the sector's radius is still determined
by the power of the node's signal).  For directional antennas one not
only has to assign a power level to each node, but also has to decide
on the direction in which each node transmits.  If the opening angle
tends to zero and the points are in general position, a strongly
connected network becomes a tour.  Hence, our results on
\tsp($2,\alpha$) may shed some light on the difficulty of power
assignment for directional antennas with small opening angles.

%%%%%%%%%%%%%%%%%%%%%%%%%%%%%%%%%%%%%%%%%%%%%%%%%%%%%%%%%%%%%%%%%%%%%%%
\section{Approximating \tspalpha}
\label{sec:tspalpha}
%%%%%%%%%%%%%%%%%%%%%%%%%%%%%%%%%%%%%%%%%%%%%%%%%%%%%%%%%%%%%%%%%%%%%%%

In this section we lay the basis for our main contribution, a
5-approximation for \tsp$(2,2)$ in Section~\ref{sec:two}.  We review
known algorithms for a related version of \tsp, which can be applied
to our setting.  As it turns out, these algorithms already yield the
same worst-case bounds as the algorithm that Funke et
al.~\cite{flnl-papwc-08} gave recently.

We recall some definitions.
% \begin{definition}[\cite{tsp_relaxed}]
  Let $S$ be a set, let $\dist(\cdot,\cdot): S \times S \rightarrow
  \R_{\ge 0}$ be a distance function on $S$, and let $\tau\geq 1$. We
  say that~$\dist(\cdot,\cdot)$ fulfills the \emph{$\tau$-relaxed
    triangle inequality} if any three elements $p,q,r\in S$ satisfy
  $\dist(p,r) \leq \tau \cdot (\dist(p,q)+\dist(q,r))$.
% \end{definition}
Recall that we denote by $\Delta_\tau$-{\tsp} the {\tsp} problem on
complete graphs whose weight function (when viewed as a distance
function on the vertices) fulfills the $\tau$-relaxed triangle
inequality.  The following lemma, which has been observed by Funke et
al.~\cite{flnl-papwc-08}, allows us to apply algorithms for
$\Delta_\tau$-{\tsp} to our problem.  The proof relies on
H\"older's inequality.

\newcommand{\mylemmaone}{%
  Let $\alpha>0$ be a fixed constant.  The distance function $| \cdot
  |^\alpha: \R^d \times \R^d \rightarrow \R_{\ge 0}, (p,q) \mapsto
  |pq|^\alpha$ fulfills the $\tau$-relaxed triangle inequality for
  $\tau=2^{\alpha-1}$.
}
\begin{lemma}[\cite{flnl-papwc-08}]
  \label{lem:relaxed}
  \mylemmaone
\end{lemma}

Andreae and Bandelt~\cite{tsp_relaxed} gave an approximation algorithm
for $\Delta_\tau$-\tsp. Their \tthree-algorithm is an adaptation of
the well-known double-spanning-tree heuristic for \tsp.  This
heuristic finds a minimum spanning tree (MST) in the given graph~$G$,
doubles all edges, finds an Euler tour in the resulting multigraph,
and finally constructs a Hamiltonian cycle from the Euler tour by
skipping all nodes that have already been visited.  The weight of the
MST is a lower bound for the length of a \tsp-tour since removing any
edge from a tour yields a spanning tree whose weight is at least the
weight of the MST.  Note that this statement holds for arbitrary
weight functions.  If the triangle inequality holds, the heuristic
yields a 2-approximation since then skipping over visited nodes never
increases the length of the tour, which initially equals twice the
weight of the MST.  For the weight function~$|\cdot|^\alpha$, however,
the heuristic can perform arbitrarily badly---consider a sequence of
$n$ equally-spaced points on a line.

The \tthree-algorithm of Andreae and Bandelt also creates a
Hamiltonian tour by shortcutting the MST, but their algorithm never
skips more than two consecutive nodes. It is never necessary to skip
more than two consecutive nodes because the cube~$T^3$ of a tree~$T$
is always Hamiltonian by a result of Sekanina~\cite{sekanina1960osv}.
Recall that the cube of a graph~$G$ contains an edge $uv$ if there is
a path from~$u$ to~$v$ in~$G$ that uses at most three edges. The proof
of Sekanina is constructive; Andreae and Bandelt use it to
construct a tour in $\text{MST}^3$.

The recursive procedure of Sekanina~\cite{sekanina1960osv} to obtain a
Hamiltonian cycle in~$T^3$ intuitively works as illustrated in
Fig.~\ref{fig:intuition_t3}; for the pseudo-code, see
Algorithm~\ref{alg:CycleInCube}.  The algorithm is applied to a
tree~$T$ and an edge~$e=u_1u_2$ of~$T$.  Removing the edge~$e$ splits
the tree into two components~$T_1$ and~$T_2$.  In each component~$T_i$
($i=1,2$), the algorithm selects an arbitrary edge $e_i=u_iw_i$
incident to~$u_i$ and recursively computes a Hamiltonian cycle
of~$T_i$ that includes the edge~$e_i$.  The algorithm returns a
Hamiltonian cycle of~$T$ that includes~$e$.  The cycle consists of the
cycles in~$T_1$ and~$T_2$ without the edges~$e_1$ and~$e_2$,
respectively.  The two resulting paths are stitched together with the
help of~$e$ and the new edge $w_1w_2$.

\begin{figure}
  \begin{minipage}{.31\textwidth}
    \centering
    \includegraphics{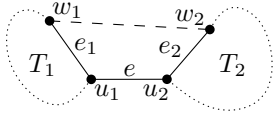}
    \caption{Recursively finding a Hamiltonian cycle in the cube of the
      tree~$T$.}
    \label{fig:intuition_t3}
  \end{minipage}
  \hfill
  \begin{minipage}{.64\textwidth}
    \begin{algorithm}[H]
      \For{$i \leftarrow 1$ \KwTo $2$}{%
        $T_i \leftarrow$ component of $T - e$ that contains $u_i$\;
        \lIf{$|T_i|=1$}{$P_i \leftarrow \emptyset$; $w_i \leftarrow u_i$\;}
        \Else{%
          pick an edge $e_i=u_iw_i$ incident to~$u_i$
          in~$T_i$\nllabel{lin:pick}\;
          \lIf{$|T_i|=2$}{$\Pi_i \leftarrow e_i$\;}
          \lElse{$\Pi_i \leftarrow$
            \textsc{CycleInCube}$(T_i, e_i) - e_i$}
        }
      }
      \Return{\normalfont $\Pi_1 + e + \Pi_2 + w_1w_2$}
      \caption{\textsc{CycleInCube}($T$, $e=u_1u_2$)}
      \label{alg:CycleInCube}
    \end{algorithm}
  \end{minipage}
\end{figure}

Note that different choices of the edge $e_i$ in line~\ref{lin:pick}
give rise to different versions of the algorithm. The standard
\tthree-algorithm takes an arbitrary such edge, while Andreae's
refined version~\cite{tsp_relaxed_improved} makes a specific choice,
which gives a better result.  (In the next section we will choose
$e_i$ based on the local geometry of the MST, which will lead to an
improved result for our problem.)  Andreae's tour in MST$^3$ has
weight at most $(\tau^2 + \tau)$ times the weight of the MST, which is
worst-case optimal~\cite{tsp_relaxed}.  Combining his result with
Lemma~\ref{lem:relaxed} yields that the refined \tthree-algorithm is a
$(4^{\alpha-1} + 2^{\alpha-1})$-approximation for \tspalpha.  We now
improve on this with the help of a simple argument.  We will
frequently use the following definition.
% \begin{definition}
  Let $T$ be a tree and let $v_0, \dots ,v_k$ be a simple path in~$T$.
  Then we call $v_0 v_k$ a \emph{$k$-shortcut} of~$T$.  We say that a
  shortcut $vw$ \emph{uses} an edge~$e$ if~$e$ lies on the path
  connecting~$v$ and~$w$ in~$T$.
% \end{definition}
It is not hard to see that
the weight of a $k$-shortcut can be bounded as follows.
\newcommand{\mylemmatwo}{%
  Let $\alpha \ge 1$ and let~$e$ be a $k$-shortcut using edges
  $e_1, \dots, e_k$. Then
  $|e|^\alpha \leq k^{\alpha-1}\sum_{i=1}^{k} |e_i|^\alpha.$
}
\begin{lemma}
  \label{lem:nshortcut}
  \mylemmatwo
\end{lemma}
Given a tree~$T$, the tour constructed by the \tthree-algorithm
consists of edges of~$T$ and 2- and 3-shortcuts that use edges of~$T$.
Note that in this tour each edge of~$T$ is used exactly twice.  Thus,
for $\alpha \ge 2$, the original \tthree-algorithm does actually
better than the bound we obtained above for the refined
\tthree-algorithm.
\begin{corollary}
  \label{cor:t3}
  Every version of the \tthree-algorithm is a
  $(2\cdot3^{\alpha-1})$-approximation for \tspalpha.
\end{corollary}
Note that our improved analysis of the \tthree-algorithm yields the
same result as the algorithm of Funke et al.~\cite{flnl-papwc-08}.

Bender and Chekuri~\cite{tsp_relaxed_improved_further} designed a
$4\tau$-approximation for $\Delta_\tau$-{\tsp} using a different
lower bound: the optimal {\tsp} tour is a biconnected subgraph of the
original graph.  The weight of the optimal {\tsp} tour is at least
that of the minimum-weight biconnected subgraph.  The latter is
\np-hard to compute~\cite{et-ap-76}, but can be approximated within a
factor of~2~\cite{ps-iaaw2-97}.  Moreover, the square of a biconnected
subgraph is always Hamiltonian.  Thus using only edges of the
biconnected subgraph and \emph{two}-shortcuts yields a
$4\tau$-approximation for $\Delta_\tau$-\tsp.  Combining the result
of Bender and Chekuri with Lemma~\ref{lem:relaxed} immediately yields
the following result, which is better than Corollary~\ref{cor:t3} for
$\alpha> \log_{3/2} 3 \approx 2.71$.
\begin{corollary}
  \label{cor:bc}
  The algorithm of Bender and Chekuri is a
  $2^{\alpha+1}$-approximation for \tspalpha.
\end{corollary}

%%%%%%%%%%%%%%%%%%%%%%%%%%%%%%%%%%%%%%%%%%%%%%%%%%%%%%%%%%%%%%%%%%%%%%%
\section{A 5-Approximation for T{\small SP}(2,2)}
\label{sec:two}
%%%%%%%%%%%%%%%%%%%%%%%%%%%%%%%%%%%%%%%%%%%%%%%%%%%%%%%%%%%%%%%%%%%%%%%

In the previous section we have used graph-theoretic arguments to
determine the performance of the \tthree-algorithm.  By
Corollary~\ref{cor:t3}, the \tthree-algorithm yields a 6-approximation
for $\alpha=2$, independently of the dimension of the underlying
Euclidean space.  We now define what we call the \emph{geometric}
\tthree-algorithm and show that it yields a 5-approximation for
\tsp$(2,2)$.  The geometric \tthree-algorithm simply chooses in
line~\ref{lin:pick} of Algorithm~\ref{alg:CycleInCube} the edge $e_i$
that makes the smallest angle with the edge~$e$.

The idea behind taking advantage of geometry is as follows.  In
Corollary~\ref{cor:t3} we have exploited the fact that each edge is
used in two $(\le 3)$-shortcuts.  The weight of a 3-shortcut is
maximum if the corresponding points lie on a line.  For the case of
the Euclidean MST it is well-known that edges make an angle of at
least $\pi/3$ if they share an endpoint.  The same proof also works
for the MST \wrt~$|\cdot|^\alpha$.  This guarantees that in
line~\ref{lin:pick} of Algorithm~\ref{alg:CycleInCube}, we can pick an
edge~$e_i$ that makes a relatively small angle with~$e$---if the
degree of~$u_i$ is larger than~2.  Otherwise, it is easy to see
that~$e_i$ is used by a ($\le 2$)- and a ($\le 3$)-shortcut, which is
favorable to being used by two 3-shortcuts, see
Lemma~\ref{lem:nshortcut}.

Although the intuition behind our geometric \tthree-algorithm is clear,
its analysis turns out to be non-trivial.
We start with the following lemma that can be proved with some
elementary trigonometry.  Given two line segments $s$ and $t$ incident
to the same point, we denote the smaller angle between $s$ and $t$ by
$\angle st$ and define $\psi_{st} = \pi - \angle st$.

\begin{figure}[tb]
  \hfill
  \subfigure[$a$ and $c$ lie on the same side of the line through~$b$]
  {\qquad\includegraphics{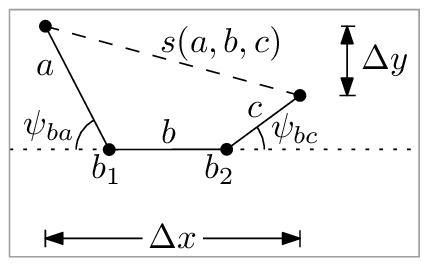}\qquad}
  \hfill
  \subfigure[$a$ and $c$ lie on different sides of the line through~$b$]
  {\qquad\includegraphics{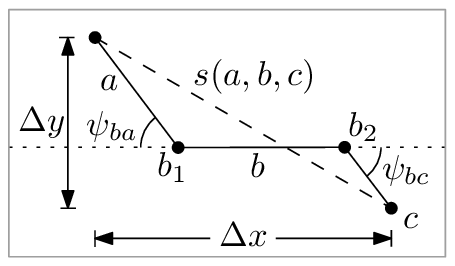}\qquad}
  \hfill\hspace*{0ex}
  \caption{Two cases for computing the length of the 3-shortcut $s(a,b,c)$.}
  \label{fig:cost3shortcut}
\end{figure}

\begin{lemma}
\label{lem:3-shortcut}
Given a tree~$T$, the 3-shortcut $s(a,b,c)$ that uses the edges~$a$, $b$, $c$
of~$T$ in this order has weight
\begin{eqnarray*}
|s(a,b,c)|^2 &=& |a|^2 + |b|^2 + |c|^2
+ 2|a||b|\cos\psi_{ba} + 2|b||c|\cos\psi_{bc}
+ 2|a||c|\cos(\psi_{ba} + \delta \cdot \psi_{bc}),
\end{eqnarray*}
where $\delta=+1$ if $a$ and $c$ lie on the same side of the line
through $b$, and $\delta=-1$ if $a$ and $c$ lie on opposite sides.
Moreover, $|s(a,b,c)|^2 \le 2|a|^2 + |b|^2 + 2|c|^2 +
2|a||b|\cos\psi_{ba} + 2|b||c|\cos\psi_{bc}$.
\end{lemma}

Lemma~\ref{lem:3-shortcut} (illustrated in
Fig.~\ref{fig:cost3shortcut}) expresses the weight of a 3-shortcut in
terms of the lengths of the edges and the angles between them.  Now we
show that if an edge~$a$ is used in two 3-shortcuts, two of these
angles are related.  Note that the \tthree-algorithm generates the two
3-shortcuts that use~$a$ in two consecutive recursive calls, see
Fig.~\ref{fig:two3shortcuts_renumbered}.  The \tthree-algorithm is
first applied to edge~$b$ and then recursively to edge~$a$.  In the
recursive call, the shortcut $s(e,a,d)$ is generated where~$d$ is an
edge incident to both~$a$ and~$b$.  Then the algorithm returns from
the recursion and generates the 3-shortcut $s(a,b,c)$.  Thus~$a$ is
the middle edge in one 3-shortcut and the first or last edge in the
other 3-shortcut.
We rely on the following.

\newcommand{\mylemmafour}{%
  If the geometric \tthree-algorithm generates the two 3-shortcuts
  $s(a,b,c)$ and $s(e,a,d)$ in two recursive calls
  and~$d$ is incident to both~$a$ and~$b$, then $\psi_{ba} \ge (\pi -
  \psi_{ad})/2$.
}
\begin{lemma}
  \label{lem:related_angles}
  \mylemmafour
\end{lemma}

Now we are ready to prove the main result of this section.

\begin{theorem}
  \label{thm:t3_5-app}
  The geometric \tthree-algorithm yields a 5-approximation for
  \tsp$(2,2)$.
\end{theorem}

\begin{proof}
  We express the length of each shortcut~$s$ of the \tthree-tour in
  terms of the lengths of the MST edges that~$s$ uses.  Changing the
  perspective, for each MST edge~$a$, we use contrib$(a)$ to denote
  the sum of all terms that contain the factor~$|a|$.  The edge~$a$ is
  used in at most two shortcuts.  Bounding their lengths yields an
  upper bound on contrib$(a)$.  The sum of all contributions relates
  the length of the \tthree-tour to that of the MST
  (\wrt~$|\cdot|^\alpha$), which in turn is a lower bound for the
  length of an optimal \tsp\ tour.

  Due to Lemma~\ref{lem:nshortcut}, contrib$(a) \le 5|a|^2$ if~$a$
  is used in a ($\le$2)-shortcut on one side and a ($\le$3)-shortcut
  on the other side.  So we focus on the case that~$a$ is used in two
  3-shortcuts, see Fig.~\ref{fig:two3shortcuts_renumbered}.
  \begin{figure}[b]
  \begin{minipage}[b]{.48\textwidth}
    \centering
    \includegraphics{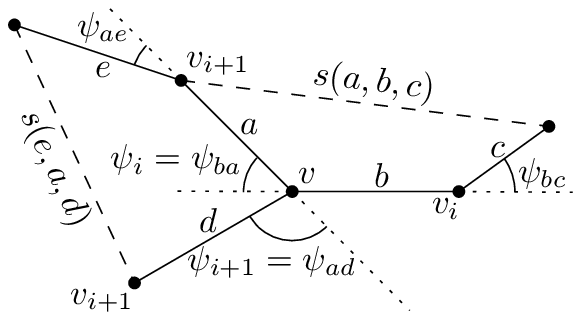}
    \caption{Two 3-shortcuts that use edge~$a$.}
    \label{fig:two3shortcuts_renumbered}
  \end{minipage}
  \hfill
  \begin{minipage}[b]{.4\textwidth}
    \centering
    \includegraphics{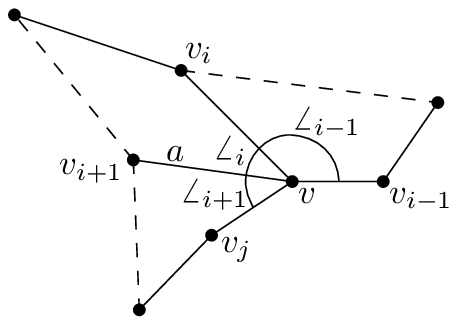}
    \caption{Illustration of case III.}
    \label{fig:three3shortcuts}
  \end{minipage}
  \end{figure}
  We rewrite the composite terms in the bound for $s(a,b,c)$ in
  Lemma~\ref{lem:3-shortcut} using Young's inequality with~$\eps$,
  which, given $x,y \in \R$ and $\eps >0$, states that $xy \le
  x^2/(2\eps) + y^2\eps/2$.

  Let~$v$ be the vertex that is incident to edges~$a$ and~$b$. If
  there are multiple 3-shortcuts that use edges that are incident
  to~$v$ then the \tthree-algorithm generates these in consecutive
  recursive calls. We renumber the edges incident to~$v$ such that the
  algorithm is first applied to $vv_1$, then recursively to $vv_2$
  etc.  Then there is some $i \ge 1$ such that $b = vv_i$ and
  $a=vv_{i+1}$ because the algorithm is first applied to~$b$ and then
  recursively to~$a$.  We define $\psi_i = \psi_{vv_i,vv_{i+1}} ( =
  \psi_{ba})$.  We rewrite the term $2|a||b|\cos\psi_{ba}$ in the
  bound for $|s(a,b,c)|^2$ in Lemma~\ref{lem:3-shortcut} as follows.
  \begin{align}
    &2|a||b|\cos\psi_{ba} = 2|vv_i||vv_{i+1}|\cos\psi_i \le
    f(|vv_{i+1}|, |vv_{i}|, \psi_i), \label{ineq:crossterms}\\
    & \text{where} \notag \\[-2ex]
    & f(|vv_{i+1}|, |vv_i|, \psi_i) =
    \begin{cases}
      0 &\text{if $\psi_i \ge \frac{\pi}{2}$,}\\
      |vv_i|^2 + |vv_{i+1}|^2\cos^2\psi_i &
      \text{if $\psi_i < \frac{\pi}{2}$ and}\\
      &\text{$\left(i=1\text{ or }\left(i>1\text{ and }\psi_{i-1} \ge
            \frac{\pi}{2}\right)\right)$,}\\
      \left(|vv_i|^2 + |vv_{i+1}|^2\right)\cos\psi_i &\text{if $\psi_i
        < \frac{\pi}{2}$ and $i>1$ and $\psi_{i-1} <
        \frac{\pi}{2}$.}\notag\\
    \end{cases}
  \end{align}
  The second case of inequality~(\ref{ineq:crossterms}) follows from
  Young's inequality with $\eps=1/\cos\psi_i$ and the third case from
  Young's inequality with $\eps=1$.  Replacing $2|b||c|\cos\psi_{bc}$
  in the bound for $|s(a,b,c)|^2$ in Lemma~\ref{lem:3-shortcut} is
  analogous.  Together, the two replacements yield the bound
  \begin{equation}
    |s(a,b,c)|^2 \le 2|a|^2 + |b|^2 + 2|c|^2 + f(|a|,|b|,\psi_{ba}) +
    f(|c|,|b|,\psi_{bc}). \label{ineq:abc_with_f}
  \end{equation}
  We use~(\ref{ineq:abc_with_f}) to bound the weights of all
  3-shortcuts.  The weight of the final tour is the sum of the weights
  of all shortcuts.  In this sum we can take the two occurrences of an
  edge $a = vv_{i+1}$ together and analyze the contribution of $a$ to
  the tour.  Note that the result of~(\ref{ineq:abc_with_f}) is still
  at most $3(|a|^2+|b|^2+|c|^2)$.  So if an edge $a$ is used in a
  ($\le$3)-shortcut on one side and a ($\le2$)-shortcut on the other
  side, then we still have that contrib$(a) \le 5|a|^2$.  It remains
  to consider the case that~$a$ is used in two 3-shortcuts.  Let
  $s(a,b,c)$ and $s(e,a,d)$ be these 3-shortcuts.  The algorithm is
  first applied to edge~$b$ and generates shortcut $s(a,b,c)$,
  where~$a$ is the first or the third edge of the shortcut.  Then the
  algorithm is recursively applied to edge~$a$ and generates shortcut
  $s(e,a,d)$, where~$a$ is the middle edge.
  Fig.~\ref{fig:two3shortcuts_renumbered} shows how the vertices are
  numbered in this case.

  Let $\sigma_a$ be a function that takes a sum of terms and returns the
  sum of all terms that contain~$|a|$.  We derive the following
  expression for contrib$(a)$.
  \begin{align}
    \text{contrib}(a) &= \sigma_a(\weight(s(a,b,c))) +
    \sigma_a(\weight(s(e,a,d)))\notag\\
    &\le \sigma_a\left(2|a|^2 + |b|^2 + 2|c|^2 + f(|a|,|b|,\psi_{ba}) +
      f(|c|,|b|,\psi_{bc})\right)\notag\\
    &\phantom{+} + \sigma_a\left(2|e|^2 + |a|^2 + 2|d|^2 +
      f(|e|,|a|,\psi_{ae}) + f(|d|,|a|,\psi_{ad})\right)\notag\\
%    &= 3|a|^2 + \sigma_a(f(|a|,|b|,\psi_{ba})) +
%    \sigma_a(f(|e|,|a|,\psi_{ae})) + \sigma_a(f(|d|,|a|,\psi_{ad}))\notag\\
%    &\le 4|a|^2 + \sigma_a(f(|a|,|b|,\psi_{ba})) +
%    \sigma_a(f(|d|,|a|,\psi_{ad}))\notag\\ =
    &\le 4|a|^2 + \sigma_a(f(|vv_{i+1}|,|vv_{i}|,\psi_i)) +
    \sigma_a(f(|vv_{i+2}|,|vv_{i+1}|,\psi_{i+1}))\label{ineq:contriba}
  \end{align}
  By definition of~$f$ we have to consider three cases in
  (\ref{ineq:contriba}) for contrib$(a)$.   \smallskip

  \noindent
  \textbf{Case I:} $\psi_i \ge \pi/2$ or $\psi_{i+1} \ge \pi/2$.

  We assume w.l.o.g.\ % without loss of generality
  that $\psi_i \ge \pi/2$.  Then
  we know that $f(|vv_{i+1}|,|vv_{i}|,\psi_i)) = 0$ and in the worst
  case $\sigma_a(f(|vv_{i+2}|,|vv_{i+1}|,\psi_{i+1})) \le |a|^2$.  Thus
  we have that contrib$(a) \le 5|a|^2.$ \smallskip

  \noindent
  \textbf{Case II:} $\psi_i < \pi/2$ and $\psi_{i+1} < \pi/2$ and ($i
  = 1$ or ($i > 1$ and $\psi_{i-1} \ge \pi/2$)).

  By definition of~$f$ we have:
\[
  \begin{array}{@{}lcccl@{}}
    \sigma_a(f(|vv_{i+1}|,|vv_{i}|,\psi_i)) &=& \sigma_a\left(|vv_i|^2 +
      |vv_{i+1}|^2\cos^2\psi_i\right) &=&
    |a|^2\cos^2\psi_i\\[.5ex]

    \sigma_a(f(|vv_{i+2}|,|vv_{i+1}|,\psi_{i+1})) &=&
    \sigma_a\left((|vv_{i+1}|^2 + |vv_{i+2}|^2)\cos\psi_{i+1}\right) &=&
    |a|^2\cos\psi_{i+1}
  \end{array}
\]
  Lemma~\ref{lem:related_angles} states that $\psi_i \ge
  (\pi-\psi_{i+1})/2$.  We also know that $\psi_i \le \pi$ by
  definition.  Thus we have
  \[\text{contrib}(a)\le\Big(4 +\cos^2\frac{\pi-\psi_{i+1}}{2} +
    \cos\psi_{i+1}\Big)|a|^2 \; \le \; 5|a|^2.\]

  \noindent\textbf{Case III:} $\psi_i < \pi/2$ and $\psi_{i+1} <
  \pi/2$ and $i>1$ and $\psi_{i-1} < \pi/2$.

  It can be shown that this leads to a contradiction, see
  Fig.~\ref{fig:three3shortcuts} (on
  page~\pageref{fig:three3shortcuts}).
\iffalse
  depicts the situation.  We use~$\angle_i$ as shorthand
  for $\angle vv_i,vv_{i+1}$.  Note that $\angle_i = \pi -
  \psi_i$.  Recall that the geometric \tthree-algorithm proceeds with
  the edge that makes the smallest angle with the current edge.  The
  algorithm is first applied to edge $vv_{i-1}$ and proceeds
  recursively with $vv_i$. So we have that $\angle_{i-1} \le
  \angle_j$ for some $j\ge i+1$. Since $\psi_{i-1} < \pi/2$ we have
  $\angle_{i-1} \ge \pi/2$ and thus also $\angle_j \ge \pi/2$. This
  leaves little room for the other angles: the smallest angle between
  $vv_i$ and $vv_j$ is at most~$\pi$. This means that at
  least one of $\angle_i$ and $\angle_{i+1}$ is at most $\pi/2$ and
  thus at least one of $\psi_i$ and $\psi_{i+1}$ is greater than
  $\pi/2$, which contradicts the assumptions of case~III.
\fi
\smallskip

  In cases~I and~II, the contribution of any edge~$|a|$ to the tour is
  at most~$5|a|^2$.  The theorem follows by summing up the
  contributions of all edges.
\end{proof}

When using the MST as a lower bound in the analysis, there is not much
room for improvement.  There are instances of \tsp(2,2) where the
\tthree-algorithm yields a tour whose weight is $4\frac{4}{11}$ times
that of the MST; see also
\cite[Theorem~4.19]{nijnatten}.

%%%%%%%%%%%%%%%%%%%%%%%%%%%%%%%%%%%%%%%%%%%%%%%%%%%%%%%%%%%%%%%%%%%%%%%
\section{Approximating \tsp$(2,\alpha)$ with $\alpha \ge 2$}
\label{sec:greatertwo}
%%%%%%%%%%%%%%%%%%%%%%%%%%%%%%%%%%%%%%%%%%%%%%%%%%%%%%%%%%%%%%%%%%%%%%%

In this section we generalize the main result of the previous section
to $\alpha \ge 2$.
Our new bound is always better than the bound $2 \cdot
3^{\alpha -1}$ of Funke et al.~\cite{flnl-papwc-08},
see also Corollary~\ref{cor:t3}.  For $\alpha < 3.41$ our bound is better than
the bound $2^{\alpha+1}$ that follows from the algorithm of Bender and
Chekuri \cite{tsp_relaxed_improved_further}, see
Corollary~\ref{cor:bc}.

\newcommand{\myobs}{%
  The function $h:[0,2\pi] \rightarrow \R, x \mapsto (2+\cos x)^k +
  (2 + \sin^2 x/2)^k$ attains its maximum value at $x=0$.  }

\begin{theorem}
  \label{thm:t3_alpha}
  The geometric T$\,^3$-algorithm yields a $(3^{\alpha-1} +
  \sqrt{6}^{\,\alpha}\!/3)$-appro\-xi\-ma\-tion for \tsp$(2,\alpha)$ if
  $\alpha \ge 2$.
\end{theorem}

\begin{proof}
  If an edge~$a$ is used in a ($\leq$2)-shortcut on one side and a
  ($\leq$3)-shortcut on the other side then the total contribution of
  $a$ to the tour is at most $(2^{\alpha-1} + 3^{\alpha-1})|a|^\alpha$
  by Lemma~\ref{lem:nshortcut}. So we will focus our analysis again on
  the case that~$a$ is used in two 3-shortcuts. For $\alpha=2$ we can
  express the weight of a 3-shortcut by Lemma~\ref{lem:3-shortcut} and
  rewrite the composite terms as in
  inequality~(\ref{ineq:crossterms}).  For $\alpha>2$ we apply
  H\"{o}lder's inequality.
  \begin{align}
    |s(a,b,c)|^\alpha &= \left(|s(a,b,c)|^{2}\right)^{\alpha/2}\notag\\
    &\leq \left(2|a|^2 + |b|^2 + 2|c|^2 + f(|a|,|b|, \psi_{ba}) +
      f(|c|, |b|, \psi_{bc})\right)^{\alpha/2}\notag\\
%    &= \big( \big(2+\sigma_a(f(|a|,|b|, \psi_{ba}))\big)|a|^2 +
%    \big(1+\sigma_b(f(|a|,|b|, \psi_{ba})) + \notag\\
%    &\quad+ \sigma_b(f(|c|, |b|,
%    \psi_{bc}))\big)|b|^2 + \big(2+\sigma_c(f(|c|, |b|,
%    \psi_{bc}))\big)|c|^2\big)^{\alpha/2}\notag\\
    &= \left(\beta_a |a|^2 + \beta_b |b|^2 + \beta_c |c|^2
    \right)^{\alpha/2}\label{ineq:beta}\\
    &\leq 3^{\alpha/2-1}\left(\beta_a^{\alpha/2} |a|^\alpha +
      \beta_b^{\alpha/2} |b|^\alpha + \beta_c^{\alpha/2} |c|^\alpha
    \right)\\[-4ex]\notag
  \end{align}
  We introduced the constants of type $\beta$ to shorten the
  expression.  
  % We applied H\"{o}lder's inequality to~(\ref{ineq:beta}) with $p=\alpha/2$, $q=\frac{\alpha/2}{\alpha/2-1}$, $n=3$, $x=(\beta_a |a|^2, \beta_b |b|^2, \beta_c |c|^2)$ and $y=(1,1,1)$.
  Note that the last inequality holds only if $\alpha>2$.

  In order to bound the contribution of an edge $a$ that is used in
  two 3-shortcuts we follow the proof of Theorem~\ref{thm:t3_5-app}.
  Since the assumptions of case~III in that proof led to a
  contradiction, it suffices to consider cases~I and~II.  \smallskip

  \noindent
  \textbf{Case I}: $\psi_i \geq \pi/2$ or $\psi_{i+1} \geq \pi/2$.
  \savespace{1}
  \begin{align*}
    \text{contrib}(a) &\leq 3^{\alpha/2-1}\left(\left(2+\cos
        \psi_i\right)^{\alpha/2} + \left(2+\cos
        \psi_{i+1}\right)^{\alpha/2}\right)|a|^\alpha\\
    &\leq3^{\alpha/2-1}\left(2^{\alpha/2} +3^{\alpha/2}\right)|a|^\alpha
    \,=\, \left(3^{\alpha-1} + \sqrt{6}^{\,\alpha}\!/3\right)|a|^\alpha
  \end{align*}

  \noindent
  \textbf{Case II}: $\psi_i < \pi/2$ and $\psi_{i+1} < \pi/2$ and ($i
  = 1$ or ($i > 1$ and $\psi_{i-1} \geq \pi/2$)).
  \savespace{1}
  \begin{align*}
    \text{contrib}(a) &\leq
    3^{\alpha/2-1}\underbrace{\left(\left(2+\cos\psi_{i+1}\right)^{\alpha/2}
        + \left(2+\sin^2 \psi_{i+1}/2\right)^{\alpha/2}\right)
    }_{\mbox{\normalsize $g_\alpha(\psi_{i+1})$}} |a|^\alpha
  \end{align*}
  Now we use the
  fact that the function $h:[0,2\pi] \rightarrow \R, x \mapsto (2+\cos
  x)^k + (2 + \sin^2 x/2)^k$ attains its maximum value at $x=0$.
  Thus $g_\alpha$ also attains its maximum in the range $[0,\pi/2)$ in
  $x=0$.  This yields
  \[\text{contrib}(a) \;\le\; 3^{\alpha/2-1} \cdot g_\alpha(0) \cdot
  |a|^\alpha \;\le\; \big(3^{\alpha-1} +
  \sqrt{6}^{\,\alpha}\!/3\big)|a|^\alpha. \] %
  In both cases we showed that contrib$(a) \le (3^{\alpha-1} +
  \sqrt{6}^{\,\alpha}\!/3)|a|^\alpha$. The theorem follows for
  $\alpha>2$ by summing up the contributions of all edges. The case
  $\alpha=2$ corresponds to Theorem~\ref{thm:t3_5-app}.
\end{proof}

%%%%%%%%%%%%%%%%%%%%%%%%%%%%%%%%%%%%%%%%%%%%%%%%%%%%%%%%%%%%%%%%%%%%%%%
\section{The Approximability of T{\small SP} and Rev-T{\small SP}}
\label{sec:revtsp}
%%%%%%%%%%%%%%%%%%%%%%%%%%%%%%%%%%%%%%%%%%%%%%%%%%%%%%%%%%%%%%%%%%%%%%%

In this section we discuss complexity and approximability of \tsp\ and
its variant Rev-\tsp, where the salesman is allowed to revisit the
cities.  Recall that for any fixed dimension $d\ge2$, \tsp$(d,1)$ is
\np-hard~\cite{papadimitriou1977ets} and admits a PTAS~\cite{arora}.

%%%%%%%%%%%%%%%%%%%%%%%%%%%%%%%%%%%%%%%%%%%%%%%%%%%%%%%%%%%%%%%%%%%%%%%
\begin{theorem}
  \tsp$(d,\alpha)$ and Rev-\tsp$(d,\alpha)$ are {\np}-hard for any
  $d\ge2$ and $\alpha>0$.
\end{theorem}
\begin{proof}
  Itai et al.~\cite{ItaiPS1982} showed that, given $n$ points in the
  unit grid, it is \np-hard to decide whether there is a {\tsp} tour
  of Euclidean length $n$.  Thus for both of our problems it is
  \np-hard to distinguish between $\opt=n$ and $\opt\ge
  n-1+\sqrt{2}^{\,\alpha}$.
\end{proof}

%%%%%%%%%%%%%%%%%%%%%%%%%%%%%%%%%%%%%%%%%%%%%%%%%%%%%%%%%%%%%%%%%%%%%%%
\begin{theorem}
  \tsp$(d,\alpha)$ and Rev-\tsp$(d,\alpha)$ are \apx-hard for any
  $d\ge3$ and any $\alpha>1$.
\end{theorem}
\begin{proof}
  We only discuss the case $d=3$ and $\alpha=2$---all other cases can
  be settled by slightly modified arguments---\tsp.
  and we only consider Rev-\tsp; a similar reduction can be used for \tsp.
  We reduce from $\{1,2\}$-\tsp, the {\tsp} on the complete
  graph where the weight of every edge is either 1 or~2; this problem
  is \apx-hard \cite{1PapadimitriouY1993}.  An instance of
  $\{1,2\}$-{\tsp} consists of the complete graph~$K_n=(V_n,E_n)$ with
  vertex set $V_n=\{v_1,\dots,v_n\}$, edge set $E_n=\{e_1,\dots,e_m\}$
  where $m=n(n-1)/2$, and edge lengths that are specified by a weight
  function $w: E_n \rightarrow \{1,2\}$.  Given $K_n$ and $w$, we
  construct a corresponding instance $P_{n,w}\subset\R^3$ of
  Rev-\tsp$(3,2)$.

  We start our construction by introducing several auxiliary line
  segments.
  For each vertex $v_i \in V_n$ we define its \emph{spine} to be the
  vertical line segment going from point $(ni,ni,n)$ to point
  $(ni,ni,nm)$.  For each edge $e_k=v_iv_j\in E_n$ with $i<j$, we
  define two corresponding line segments that are parallel to the
  $xy$-plane and that are called \emph{bones}.  The first bone
  connects point $(ni,ni,nk)$ on the spine of~$v_i$ to the point
  $(nj,ni,nk)$.  The other bone connects point $(nj,nj,nk)$ on the
  spine of~$v_j$ to the point $(nj,ni-\delta_k,nk)$, where
  $\delta_k=1$ if $w(e_k)=1$ and $\delta_k=\sqrt{2}$ if $w(e_k)=2$.
  Note that these two bones do not quite touch; they are separated by
  a gap of length~$\delta_k$.

  In order to get the instance~$P_{n,w}$ of Rev-\tsp$(3,2)$, we
  subdivide every single (spine or bone) line segment introduced above
  by a dense, evenly distributed set of points---we call these points
  \emph{cities} from now on---so that every unit-length piece
  receives~$n^5$ cities.  The distance between adjacent cities is
  $1/n^5$, and so the cost for going from one city to an adjacent city
  is~$1/n^{10}$.  All these cities together form instance~$P_{n,w}$,
  and this completes our construction.  Since we have introduced line
  segments with a total length of at most $n\cdot n(m-1) +
  m\cdot2n(n-1)<2n^4$, the overall number of cities is at most $2n^9$.

  For $1\le i\le n$ we call the cities on the spine of~$v_i$ and on
  all bones incident to this spine the \emph{city cluster} of~$v_i$.
  Traversing all cities within such a city cluster is very cheap; even
  if we visit every city twice, this costs at most $2 \cdot
  2n^9/n^{10}=4/n$ for all cities in all city clusters together.  In a
  traveling salesman tour, the only expensive steps occur when the
  salesman jumps from one city cluster to another city cluster.  By
  the above definition of~$\delta_k$, when jumping from bone to bone
  across the gap corresponding to edge~$e_k$ the incurred cost is
  exactly~$w(e_k)$.  Note that jumping from city cluster to city
  cluster in any other way would be much more expensive and would thus
  not reduce the total cost of the tour.

  Finally, let us show that our reduction is approximation preserving.
  Fix an $\eps$ with $0<\eps<1$.  Consider an instance $K_n$ and $w$
  of $\{1,2\}$-{\tsp}, and assume without loss of generality that
  $n>4/\eps$.  Consider an optimal tour $\pi_0$ for this instance.  If
  $\pi_0$ uses $\ell\ge0$ edges of length~$2$ and $n-\ell$ edges of
  length~$1$, then it has cost $n+\ell$.  Given a PTAS for Rev-{\tsp},
  we show how to compute in polynomial time a tour of cost at most
  $(1+\eps)(n+\ell)$ for $K_n$ and $w$.

  First note that the tour~$\pi_0$ can be transformed into a
  tour~$\pi_1$ through~$P_{n,w}$ that makes~$\ell$ jumps of cost~2 and
  $n-\ell$ jumps of cost~1.  That tour $\pi_1$ costs at most $n+\ell +
  4/n$.  Using our hypothetical PTAS for Rev-\tsp, we can compute for
  any $\eps'>0$ in polynomial time a tour~$\pi_2$ through~$P_{n,w}$ of
  cost at most $(1+\eps')c_\mathrm{opt}$, where~$c_\mathrm{opt}$ is
  the cost of an optimal Rev-\tsp\ tour.  The existence of~$\pi_1$
  yields $c_\mathrm{opt}\le n+\ell + 4/n$.  The tour~$\pi_2$ can be
  transformed into a tour~$\pi_3$ through $K_n$: Just map the jumps
  of~$\pi_2$ to the corresponding edges of~$K_n$.  Since this mapping
  cannot increase the cost, tour~$\pi_3$ costs at most
  $(1+\eps')(n+\ell + 4/n)$.  Choosing~$\eps'=\eps/2$ and using
  $4/n<\eps<1$, we can bound the cost of~$\pi_3$ from above by
  \begin{eqnarray*}
    \left(1+\frac{\eps}{2}\right)(n+\ell)+\left(1+\frac{\eps}{2}\right)\eps
    &=&   \left(1+\frac{\eps}{2}\right)(n+\ell)+\frac{\eps}{2}(2+\eps)
    ~=~ (1+\eps)(n+\ell)
  \end{eqnarray*}
  as desired.  Like~$\pi_2$, the tour~$\pi_3$ may visit vertices more
  than once.  This can be fixed by greedily introducing shortcuts.  The
  shortcuts do not increase the cost of the tour since the weight
  function~$w$ (trivially) fulfills the triangle inequality.
\end{proof}

\begin{theorem}
  \label{thm:ptas}
  There exists a PTAS for Rev-\tsp$(2,\alpha)$ for any $\alpha\ge 2$.
\end{theorem}
\begin{proof}
  Given a set $P$ of points in the plane, consider the \emph{Gabriel
    graph} $G_P$ that has a vertex for each point in~$P$.  There is an
  edge between points~$p$ and~$q$, if the open disk with diameter $pq$
  is empty, in other words, if for all points $r \in P \setminus
  \{p,q\}$, the angle $\angle prq$ is at most $\pi/2$.  The weight of
  the edge is $|pq|^\alpha$.  Note that $|pr|^\alpha+|rq|^\alpha \le
  |pq|^\alpha$ if $\angle prq$ is at least $\pi/2$.  Therefore, there
  is an optimal {\tsp} tour with revisits through~$P$ that only uses
  the edges of $G_P$: Indeed, if a tour uses an edge $pq$ for which
  there is a point~$r$ with $\angle prq \geq \pi/2$, then replacing
  $pq$ by $pr$ and $rq$ would shorten the tour.  Such a replacement is
  feasible since revisiting city~$r$ is allowed.
  The Gabriel graph is planar.  Hence we end up
  with an instance of the {\tsp} on weighted planar graphs, for which
  a PTAS is known~\cite{AroraGKKPW1998}.
\end{proof}

Recall that a \emph{quasi-}PTAS is an approximation scheme with
running time $n^{\mathrm{polylog}\;n}$, where $n$ is the size of the
input.  The following result follows immediately from the facts that
(a)~the metric $|\cdot|^\alpha$ has bounded doubling dimension and
(b)~\tsp\ on metrics of bounded doubling dimension admits a
quasi-PTAS~\cite{Talwar2004}.

\begin{theorem}
  There exists a quasi-PTAS for Rev-\tsp$(d,\alpha)$ for any
  $\alpha\in (0,1]$ and % any
  $d\ge1$.
\end{theorem}

%%%%%%%%%%%%%%%%%%%%%%%%%%%%%%%%%%%%%%%%%%%%%%%%%%%%%%%%%%%%%%%%%%%%%%%
\section{Conclusions}
%%%%%%%%%%%%%%%%%%%%%%%%%%%%%%%%%%%%%%%%%%%%%%%%%%%%%%%%%%%%%%%%%%%%%%%

In order to construct considerably better approximation algorithms for
\tsp$(d,\alpha)$, we expect that substantially different methods of
analysis have to be found.
A result of Van Nijnatten \cite[Theorem~4.19]{nijnatten} indicates that
there is not much room left for improvement as long as we compare to
the MST.

The approximability of Rev-\tsp$(2,\alpha)$ for $1<\alpha<2$ is an
interesting open question.
We believe that a (quasi)-PTAS may be
obtained using the framework of the PTAS for weighted planar graph
\tsp\ by Arora et al.~\cite{AroraGKKPW1998}. A simple reduction shows
that deriving a PTAS for our problem is at least as hard as deriving a PTAS for
weighted planar graph \tsp. Assume we have a PTAS for
Rev-\tsp$(2,\alpha)$ for some $\alpha>1$. Given a weighted planar
graph and a planar embedding, we replace each edge by a dense set of
points such that traversing
a subedge basically costs zero.  By making one subedge
of each edge~$e$ longer, we can make the cost of that subedge (and
thus of~$e$) in Rev-\tsp\ proportional to the weight of~$e$.
Then, the costs of the optimal solutions of the two problems will be
the same up to an arbitrarily
small constant factor of $1+\eps$. Such a reduction is polynomially
bounded if all weights are polynomially bounded,
which can be achieved by a standard rounding scheme.

A PTAS for Rev-\tsp$(2,\alpha)$ for any $\alpha>1$ would be an
interesting generalization of the existing PTAS's for weighted planar
graphs. Ideally, one would have a PTAS with running time independent
of $\alpha$ since it would contain both Euclidean \tsp\ and weighted
planar graph \tsp\ as special cases.

\end{document}